\documentclass[10pt,onecolumn,twoside]{IEEEconf}      

\usepackage[utf8x]{inputenc}
\usepackage{amsmath}

\usepackage[normalem]{ulem}
\usepackage{xcolor}
\usepackage{amssymb}
\usepackage{amsfonts}
\usepackage{exscale}
\usepackage{graphicx}
\usepackage{booktabs}
\usepackage{authblk}
\usepackage{float}
\usepackage[makeroom]{cancel}
\usepackage{setspace}
\usepackage{soul}
\usepackage{epstopdf}
\makeatletter
\@addtoreset{chapter}{part}
\makeatother  
\usepackage{caption}
\usepackage{subcaption}
 \usepackage{kantlipsum,cuted}
\usepackage{cite}

\newtheorem{theorem}{Theorem}[section]
\newtheorem{lemma}[theorem]{Lemma}

\newtheorem{proposition}[theorem]{Proposition}

\newtheorem{definition}[theorem]{Definition}

\newtheorem{remark}[theorem]{Remark}
\newtheorem{assumption}[theorem]{Assumption}
\numberwithin{equation}{section}

\def\begequarr{\begin{eqnarray*}}
\def\endequarr{\end{eqnarray*}}
\def\begarr{\begin{array}}
\def\endarr{\end{array}}
\def\begequ{\begin{equation}}
\def\endequ{\end{equation}}

\def\begdes{\begin{description}}
\def\enddes{\end{description}}
\def\begenu{\begin{enumerate}}
\def\begite{\begin{itemize}}
\def\endite{\end{itemize}}
\def\endenu{\end{enumerate}}

\def\lef[{\left[\begin{array}}
\def\rig]{\end{array}\right]}

\def\begcen{\begin{center}}
\def\endcen{\end{center}}

\newcommand{\be}{\begin{equation}}
\newcommand{\ee}{\end{equation}}

\newcommand{\bk} {\color{black} }

\newcommand{\bal} {\left( \begin{array}}
\newcommand{\bals} {\left[ \begin{array}}
\newcommand{\ear} {\end{array}\right) }
\newcommand{\ears} {\end{array}\right] }

\def\begcen{\begin{center}}
\def\endcen{\end{center}}
\newcommand{\col}{ \mbox{col} }

\usepackage{empheq}

\def\begs{\begin{split}}

\def\C{\mathrm{C}}

\def\L2{{\cal L}_2}
\def\L2e{{\cal L}_{2e}}

\def\rea{\mathbb{R}}

\def\diag{\mbox{diag}}

\def\col{\mbox{col}}

\def\begequarr{\begin{eqnarray}}
\def\endequarr{\end{eqnarray}}
\def\begequarrs{\begin{eqnarray*}}
\def\endequarrs{\end{eqnarray*}}
\def\begarr{\begin{array}}
\def\endarr{\end{array}}
\def\begequ{\begin{equation}}
\def\endequ{\end{equation}}

\def\begdes{\begin{description}}
\def\enddes{\end{description}}
\def\begenu{\begin{enumerate}}
\def\begite{\begin{itemize}}
\def\endite{\end{itemize}}
\def\endenu{\end{enumerate}}

\def\lef[{\left[\begin{array}}
\def\rig]{\end{array}\right]}

\def\begcen{\begin{center}}
\def\endcen{\end{center}}
\def\begrem{\begin{remark}\rm}
\def\endrem{\end{remark}}

\title{\LARGE \bf A tool for stability and power sharing analysis of a generalized class of droop controllers for high-voltage direct-current transmission systems}

\author{Daniele Zonetti, Romeo Ortega and Johannes Schiffer
 \thanks{D. Zonetti and R. Ortega are with the Laboratoire des Signaux et Syst\'{e}mes, 3, rue Joliot Curie, 91192 Gif-sur-Yvette, France.
        {\tt\small daniele.zonetti@gmail.com, romeo.ortega@lss.supelec.fr}}
        \thanks{J. Schiffer is with the School of Electronic and Electrical Engineering, University of Leeds, Leeds LS2 9JT, UK,
        {\tt\small j.schiffer@leeds.ac.uk}}
}

\begin{document}

\maketitle
\thispagestyle{empty}
\pagestyle{empty}

\begin{abstract} The problem of primary control of high-voltage direct current transmission systems is addressed in this paper, which contains four main contributions. First, to propose a new nonlinear, more realistic, model for the system suitable for primary control design, which takes into account nonlinearities introduced by conventional inner controllers. Second, to determine necessary conditions---dependent on some free controller tuning parameters---for the existence of equilibria. Third, to formulate additional (necessary) conditions for these equilibria to satisfy the power sharing constraints. Fourth, to establish conditions for stability of a given equilibrium point. The usefulness of the theoretical results is illustrated via numerical calculations on a four-terminal example.

\end{abstract}

\section{INTRODUCTION}
 For its correct operation, high-voltage direct current (hvdc) transmission systems---like all electrical power systems---must satisfy a large set of different regulation objectives that are, typically, associated to the multiple time–-scale behavior of the system. One way to deal with this issue, that prevails in practice, is the use of hierarchical control architectures \cite{Egeaalvarez,vrana,andreasson}. Usually, at the top of this hierarchy, a centralized controller called \textit{tertiary control}---based on power flow optimization algorithms (OPFs)---is in charge of providing  the inner controllers with the operating point to which the system has to be driven, according to technical and economical constraints \cite{Egeaalvarez}. If the tertiary control had exact knowledge of such constraints and of the desired operating points of all terminals, then it would be able to formulate a nominal optimization problem and  the lower level   (also called \textit{inner-loop})  \bk controllers could operate under \textit{nominal conditions}. However, such exact knowledge of all system parameters is impossible in practice, due to uncertainties and lack of information. Hence, the operating points generated by the tertiary controller may, in general, induce unsuitable \textit{perturbed conditions}. To cope with this problem further control layers, termed \textit{primary} and \textit{secondary control}, are introduced. These take action---whenever a perturbation occurs---by promptly adjusting the references provided by the tertiary control in order to preserve properties that are essential for the correct and safe operation of the system. The present paper focuses on the primary control layer. Irrespectively of the perturbation and in addition to ensuring stability, primary control has the task of preserving two fundamental criteria: a prespecified power distribution (the so-called \textit{power sharing}) and keeping the terminal voltages near the nominal value \cite{Beerten2013}. Both objectives are usually achieved by an appropriate control of the dc voltage of one or more terminals at their point of interconnection with the hvdc network \cite{sun,haile,vrana}.\\
Clearly, a \textit{sine qua non} requirement for the fulfillment of these objectives is the existence of a \textit{stable equilibrium point} for the perturbed system. The ever increasing use of power electronic devices in modern electrical networks, in particular the presence of  \textit{constant power devices} (CPDs), induces a highly nonlinear behavior in the system---rendering  the analysis of existence and stability of equilibria very complicated.  Since linear, inherently stable, models, are usually employed for the description of primary control of dc grids \cite{haile,andreasson,zhao}, little attention has been paid to the issues of stability and existence of equilibria. This fundamental aspect of the problem has only recently attracted the attention of power systems researchers \cite{depersiscpl,simpsoncpl,barabanov} who, similarly to the present work, invoke tools of nonlinear dynamic systems analysis, to deal with the intricacies of the actual nonlinear behavior.\\

The main contributions and the organization of the paper are as follows. Section~\ref{sec_modeling} is dedicated to the formulation---under some reasonable assumptions---of a reduced, nonlinear model of an hvdc transmission system in closed-loop with standard \textit{inner-loop} controllers.    In Section~\ref{sec_short} a further model simplification, which holds for a general class of dc systems with short lines configurations, is  presented.   A first implication is that both obtained  models, which are \textit{nonlinear}, may in general have no equilibria. \bk  Then, we consider a generalized class of primary controllers, that includes the special case of the ubiquitous \textit{voltage droop control}, and establish necessary conditions on the control parameters for the existence of an equilibrium point. This is done in Section~\ref{sec_existence}. An extension of this result to the problem of existence of equilibria that verify the power sharing property is  carried out in Section~\ref{sec_powersharing}. A last contribution is provided in Section~\ref{sec_stability}, with a (local) stability analysis of a \textit{known} equilibrium point, based on Lyapunov's first method. The usefulness of the theoretical results is illustrated with a numerical example in Section~\ref{sec_example}. We wrap-up the paper by drawing some conclusions and providing guidelines for future investigation.\\

\noindent \textbf{Notation.} For a set $\mathcal{N}=\{l,k,\dots,n\}$ of, possibly unordered, elements,  we denote with $i\sim\mathcal{N}$ the elements \mbox{$i=l,k,\dots,n$}. All vectors are column vectors. Given positive integers $n$, $m$, the symbol  $ 0_n\in\rea^n$ denotes the vector of all zeros,  $  0_{n\times m}$ the $n \times m$ column matrix of all zeros, \mbox{$\mathsf{1}_n\in\rea^n$} the vector with all ones and $\mathbb{I}_n$ the $n \times n$ identity matrix. When clear from the context dimensions are omitted and vectors and matrices introduced above are simply denoted by the symbols $0$, $\mathsf{1}$ or $\mathbb{I}$. For a given matrix $A$, the $i$-th colum is denoted by $A_i$. Furthermore, $\diag\{a_i\}$ is a diagonal matrix with entries $a_i \in \rea$ and $\text{bdiag}\{A_i\}$ denotes a block diagonal matrix with matrix-entries $A_i$. $x:=\col(x_1,\dots,x_n)\in\rea^n$ denotes a vector with entries $x_i \in \rea$. When clear from the context it is simply referred to as  $x:=\col(x_i)$. 
\section{NONLINEAR MODELING OF HVDC TRANSMISSION SYSTEMS}\label{sec_modeling}
 \subsection{A graph description}
The main components of an hvdc transmission system are ac to dc power converters and dc transmission lines. The power converters connect ac subsystems---that are associated to renewable generating units or to ac grids---to an hvdc network. In \cite{zonettiCEP} it has been shown that an hvdc transmission system can be represented by a directed
graph\footnote{A directed graph is an ordered 3-tuple, $\mathcal{G}^\uparrow=\{\mathcal{\mathcal N,\mathcal E},\Pi\}$,
consisting of a finite set of nodes $\mathcal{N}$, a finite set of directed edges $\mathcal{E}$ and a mapping $\Pi$ from $\mathcal{E}$ to the set of ordered pairs of $\mathcal{N}$.} without self-loops, where the power units---\textit{i.e.} power converters and transmission lines---correspond to edges and the buses correspond to nodes. Hence, a first step towards the construction of a suitable model for primary control analysis and design is then the definition of an appropriate graph description of the system topology that takes into account the primary control action.\\

We consider an hvdc transmission system described by a graph $\mathcal{G}^\uparrow(\mathcal{N},\mathcal{E})$, where $\mathsf{n}=\mathsf{c}+1$ is the number of nodes,   where the additional node is used to model the ground node, \bk   and $\mathsf{m}=\mathsf{c}+\mathsf{t}$ is the number of  edges, with $\mathsf{c}$ and $\mathsf{t}$ denoting the number of converter and transmission units respectively.   We implicitly assumed that transmission (interior) buses are eliminated  via Kron reduction \cite{doerfler13_3}. \bk  We further denote by $\mathsf{p}$ the number of converter units not equipped with primary control---termed \textit{PQ units} hereafter---and by $\mathsf{v}$ the number of converter units equipped with primary control---that we call \textit{voltage-controlled units},   with $\mathsf{c}=\mathsf{p}+\mathsf{v}.$ \bk  To facilitate reference to different units we find it convenient to partition the set of   converter   nodes (respectively converter edges) into two ordered subsets $\mathcal{N}_P$ and $\mathcal{N}_V$ (respectively $\mathcal{E}_P$ and $\mathcal{E}_V$)  corresponding  to $PQ$ and voltage-controlled nodes (respectively edges). The incidence matrix associated to the graph is given by:
\begin{equation}\label{incidencedc2}
\mathcal{B} =\begin{bmatrix}
\mathbb{I}_{\mathsf{p}} &0&\mathcal{B}_{{P}} \\
0&\mathbb{I}_{\mathsf{v}} & \mathcal{B}_{{V}} \\
-\mathsf{1}^\top_{\mathsf{p}} &-\mathsf{1}^\top_{\mathsf{v}} & 0 
 \end{bmatrix} \in\mathbb{R}^{ \mathsf{n}\times  \mathsf{m}},
\end{equation}
where the submatrices $\mathcal{B}_{{P}} \in\mathbb{R}^{\mathsf{p}\times\mathsf{t}} $ and $\mathcal{B}_{{V}} \in\mathbb{R}^{\mathsf{v}\times\mathsf{t}} $ fully capture the topology of the hvdc network with respect to the different units.

  \subsection{Converter units} \bk
For    a characterization   of the converter units we consider power converters based on voltage source converter (VSC) technology \cite{iravani}.   Since this paper focuses on primary control, we first provide a description of a single VSC  in closed-loop with the corresponding \textit{inner-loop} controller.   In hvdc transmission systems, the inner-loop controller  is usually achieved via a cascaded control scheme consisting of a current control loop whose setpoints are specified by an outer power loop \cite{colethesis}. \bk Moreover, such a control scheme  employs a phase-locked-loop (PLL) circuit, which is a circuit that synchronizes an oscillator with a reference sinusoidal input \cite{PLLbook}.   The  PLL is thus locked to the phase $a$ of the voltage   $v_{ac,i}(t)$    and allows,   under the assumption of balanced operation of the phases, to express the model in a suitable $dq$ reference frame, upon which the current and power loops are designed, see \cite{zonetti2016energy}, \cite{blasko} for more details on this topic.  For these layers of control, different strategies can be employed in practice. Amongst these, a technique termed  \textit{vector control} that consists of combining feedback linearization and PI control  is very popular, see  \cite{xuinner,blasko,lee} for an extensive overview on this control strategy. A schematic description of the VSC and of the overall control architecture, which also includes,   if any,    the primary control layer, is given in   Fig.~\ref{vectorcontrol}. \bk
\begin{figure*} 
 \centering
 \includegraphics[width=0.8\linewidth]{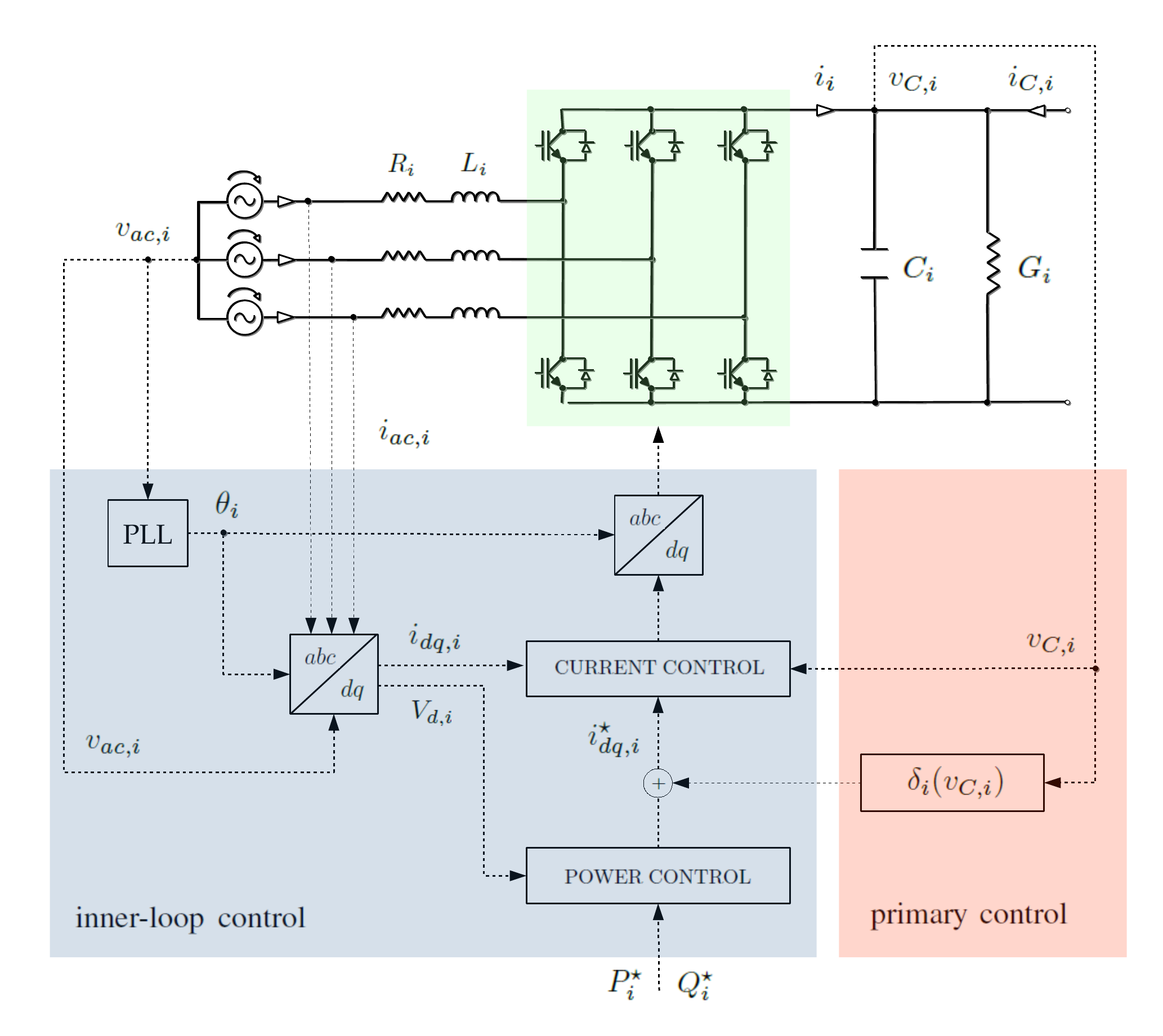}
 \caption{  Control architecture of a three-phase voltage source converter that interfaces an ac subsystem---characterized by a three--phase input voltage $v_{ac,i}(t)$---to an hvdc network---characterized by an ingoing dc current $i_{C,i}(t)$.   Bold lines represent electrical connections, while dashed lines represent signal connections \cite{zonetti2016energy}. \bk }
 \label{vectorcontrol}
\end{figure*}
 As detailed above, the \textit{inner-loop} control scheme is based on an appropriate $dq$ representation of the ac-side dynamics of the VSC, which for balanced operating conditions is given by the following second order dynamical system \cite{blasko}:
\begin{equation}\label{vsc-currents}
\begin{aligned}
L_i\dot I_{d,i}&=-R_iI_{d,i}+L_i\omega_iI_{q,i}-d_{d,i}v_{C,i}+V_{d,i}\\
L_i\dot I_{q,i}&=-L_i\omega_iI_{d,i}-R_iI_{q,i}-d_{q,i}v_{C,i}+V_{q,i}\\
\end{aligned}
\end{equation}
where $I_{d,i}\in\mathbb{R}$ and $I_{q,i}\in\mathbb{R}$ denote the direct and quadrature currents, $v_{C,i}\in\mathbb{R}_+$ denotes the dc voltage, $d_{d,i}\in\mathbb{R} $ and $d_{q,i}\in\mathbb{R}$ denote the direct and quadrature duty ratios,  $V_{d,i}\in\mathbb{R}$ and $V_{q,i}\in\mathbb{R}$ denote the direct and quadrature input voltages, $L_i\in\mathbb{R}_+$ and $R_i\in\mathbb{R}_+$ denote  the (balanced) inductance and the resistance respectively. Moreover, the dc voltage dynamics  can be described by the following scalar dynamical system:
\begin{equation}\label{vsc-voltage}
C_i\dot v_{C,i}=-G_iv_{C,i}+ i_i+i_{C,i},\quad i_{i}:=d_{d,i}I_{d,i} +d_{q,i} I_{q,i},
\end{equation}
where $i_{C,i}\in\mathbb{R}$ denotes the current coming from the dc network, $i_{i} $ denotes the dc current injection via the VSC, $C_i\in\mathbb{R}_+$ and $G_i\in\mathbb{R}_+$ denote the capacitance and the conductance respectively. For a characterization of the power injections we  consider the standard definitions of instantaneous active and reactive power associated to the ac-side of the VSC, which are given by \cite{akagi,teodorescu}:
\begin{equation}\label{powers}
\begin{aligned}
P_i   &:=V_{d,i}   I_{d,i}   +V_{q,i}   I_{q,i} , \quad Q_i   :=V_{q,i}  I_{d,i}   -V_{d,i}  I_{q,i},
\end{aligned}
\end{equation}
while the dc power associated to the dc-side is given by:
\begin{equation}\label{dcpower}
P_{\mathrm{DC},i}:=v_{C,i}i_i.
\end{equation}
We now make two standard assumptions on the design of the \textit{inner-loop} controllers.
\smallbreak
 \begin{assumption}\label{Vq0} $V_{q,i}=V_{q,i}^\star=0,\quad \forall t\geq0.$ 
\end{assumption}
\smallbreak
\begin{assumption}\label{assCfast} All \textit{inner-loop} controllers are characterized by stable current control schemes. Moreover, the employed schemes guarantee instantaneous and exact tracking of the desired currents.
\end{assumption}
\smallbreak
Assumption \ref{Vq0} can be legitimized  by appropriate design of the PLL mechanism, which is demanded to fix the $dq$ transformation angle so that the quadrature voltage is always kept  zero after very small transients. Since a PLL usually operates in a range of a few $ms$, which is smaller than the time scale at which the power loop evolves, these transients can be neglected.\\
Similarly, Assumption \ref{assCfast} can be legitimized by an appropriate design of the current control scheme so that the resulting closed-loop system is internally stable and has a very large bandwidth compared to the dc voltage dynamics and to the outer loops. In fact, tracking of the currents is usually achieved in $10-50$ $ms$, while dc voltage dynamics and outer loops evolve at a much slower time-scale \cite{Egeaalvarez}.\\

 Under Assumption \ref{Vq0} and Assumption \ref{assCfast}, from the stationary equations of the currents dynamics expressed by \eqref{vsc-currents}, \textit{i.e.} for $\dot{I}_{d,i}^\star=0$, $\dot{I}_{q,i}^\star=0$, we have that
 \begin{equation}\label{ddq}
 \begin{aligned}
d_{d,i}^\star&=\frac{1}{v_{C,i}}\left(-R_iI_{d,i}^\star+L_i\omega_i I_{q,i}^\star+V_{d,i}^\star\right),\quad
d_{q,i}^\star=\frac{1}{v_{C,i}}\left(-L_i\omega_iI_{d,i}^\star-R_i I_{q,i}^\star\right),
\end{aligned}
 \end{equation}
 where $I_{d,i}^\star$ and $I_{q,i}^\star$ denote the controlled $dq$ currents (the dynamics of which are neglected under Assumption \ref{assCfast}), while $V_{d,i}^\star$ denotes the corresponding direct voltage on the ac-side of the VSC. By substituting \eqref{ddq} into \eqref{vsc-voltage} and recalling the definition of active power provided in \eqref{powers}, the controlled dc current can thus be expressed as
  \begin{equation}\label{injectedc}
 i_i^\star=\frac{V_{d,i}^\star I_{d,i}^\star-R_i(I_{d,i}^\star)^2-R_i(I_{q,i}^\star)^2}{v_{C,i}}=\frac{P_i^\star-D_i^\star}{v_{C,i}},
 \end{equation}
where
\begin{equation}\label{PD}
P_i^\star:=V_{d,i}^\star I^\star_{d,i},\quad D_i^\star:=R_i\left[(I_{d,i}^\star)^2+(I_{q,i}^\star)^2\right]
\end{equation}
denote respectively the controlled active power on the ac-side and the power dissipated internally by the converter. We then make a further assumption.
\smallbreak
\begin{assumption}\label{assdiss} $D_i^\star=0.$
\end{assumption}
\smallbreak
Assumption \ref{assdiss} can be justified by the high efficiency of the converter, \textit{i.e.} by the small values of the balanced three-phase resistance $R$, which yield $D_i^\star\approx 0$. Hence, by replacing \eqref{injectedc} into \eqref{vsc-voltage} and using the definitions \eqref{PD}, we obtain the following scalar dynamical system \cite{teodorescu}:
 \begin{equation}\label{Vdyn0}
 C_i\dot v_{C,i}=-G_iv_{C,i}+\frac{V^\star_{d,i}}{v_{C,i}}I^\star_{d,i}+i_{C,i}
 \end{equation}
with $i\sim\mathcal{E}_P\cup\mathcal{E}_V$, which describes the dc-side dynamics of a VSC under assumptions \ref{Vq0}, \ref{assCfast} and \ref{assdiss}. By taking \eqref{Vdyn0} as a point of departure, we   next   derive the dynamics of the current-controlled VSCs in closed-loop with the outer power control. \smallbreak

If the unit is a \textit{PQ unit}, the current references are simply determined by the outer power loop via \eqref{powers}   with constant active power $P_j^\mathrm{ref}$ and reactive power $Q_j^\mathrm{ref}$, which by noting that $V_{q,j}^\star=0$, are given by:
\begin{equation}\label{refsPQ}
I_{d,j}^\star  =\frac{P^{\mathrm{ref}}_j}{V^\star_{d,j}  },\quad I_{q,j}^\star =-\frac{Q^\mathrm{ref}_j}{V^\star_{d,j}   },
\end{equation}
with $j\sim\mathcal{E}_P$, which replaced into \eqref{Vdyn0} gives
 \begin{equation}\label{VdynPQ}
 C_j\dot v_{C,j}=-G_jv_{C,j}+u_j(v_{C,j})+i_{C,j}.
 \end{equation}
with the new current variable $u_j$ and the dc voltage $v_{C,j}$ verifying the hyperbolic constraint $P^\mathrm{ref}_{j}= v_{C,j}u_j$,  $j\sim\mathcal{E}_P$. Hence, a \textit{PQ unit} can be approximated, with respect to its power behavior, by a constant power device of value $P^\mathrm{ref}_{P,j}:=P^\mathrm{ref}_j$, see also Fig.~\ref{approx1}. On the other hand, if the converter unit is a \textit{voltage-controlled unit}, the current references are   modified    according to the primary control strategy. A common approach in this scenario is to introduce an additional deviation (also called \textit{droop}) in the direct current reference---obtained from the outer power loop---as a function of the dc voltage, while keeping the calculation of the reference of the quadrature current unchanged:
\begin{equation}\label{refscurrentZIP}
I_{d,k}^\star=\frac{P_k^\mathrm{ref}}{V_{d,k}^\star}+\delta_k(v_{C,k}),\quad 
I_{q,k}^\star=-\frac{Q_k^\mathrm{ref}}{V_{d,k}^\star},
\end{equation}
with $k\sim\mathcal{E}_V$ and where $\delta_k(v_{C,k})$ represents the state-dependent contribution provided by the primary control.  We propose the primary control law:
\begin{equation}\label{refsZIP}  
\delta_k(v_{C,k} )= \frac{1}{V_{d,k}^\star}\left(\mu_{P,k} +\mu_{I,k}v_{C,k}+\mu_{Z,k}v^2_{C,k}\right),  
\end{equation}
with $k\sim\mathcal{E}_V$ and where $\mu_{P,k}$, $\mu_{I,k}, \mu_{Z,k} \in\mathbb{R}$ are free control parameters. By replacing \eqref{refscurrentZIP}-\eqref{refsZIP} into \eqref{Vdyn0}, we obtain
 \begin{equation}\label{VdynZIP}
 C_k\dot v_{C,k}=-(G_k-\mu_{Z,k})v_{C,j}+\mu_{I,k}+u_k(v_{C,k})+i_{C,k},
 \end{equation}
with the new current variable $u_k$ and the dc voltage $v_{C,k}$ verifying the hyperbolic constraint $P^\mathrm{ref}_{k}+\mu_{P,k}=v_{C,k}u_k$, $k\sim\mathcal{E}_V$. Moreover, with Assumption \ref{assdiss} the injected dc power is given by:
 \begin{equation}\label{DCinjected}
 P_{\mathrm{DC},k}(v_{C,k})= P^\mathrm{ref}_{V,k}+\mu_{I,k}v_{C,k}+\mu_{Z,k}v_{C,k}^2,
 \end{equation}
with
$$
P^\mathrm{ref}_{V,k}:=P_k^\mathrm{ref}+\mu_{P,k},
$$
from which follows, with the control law \eqref{refsZIP},  that a \textit{voltage-controlled unit} can be approximated, with respect to its power behavior, by  a ZIP model, \textit{i.e.}  the parallel connection of a constant impedance (Z), a constant current source/sink (I) and a constant power device (P). More precisely---see also Fig.~\ref{approx2}---the parameters $P^\mathrm{ref}_{{V},k}$, $\mu_{I,k}$ and $\mu_{Z,k}$ represent the constant power, constant current and constant impedance of the  ZIP model.\bk
\begin{figure*}
\begin{subfigure}{0.44\textwidth}
\includegraphics[width=\linewidth]{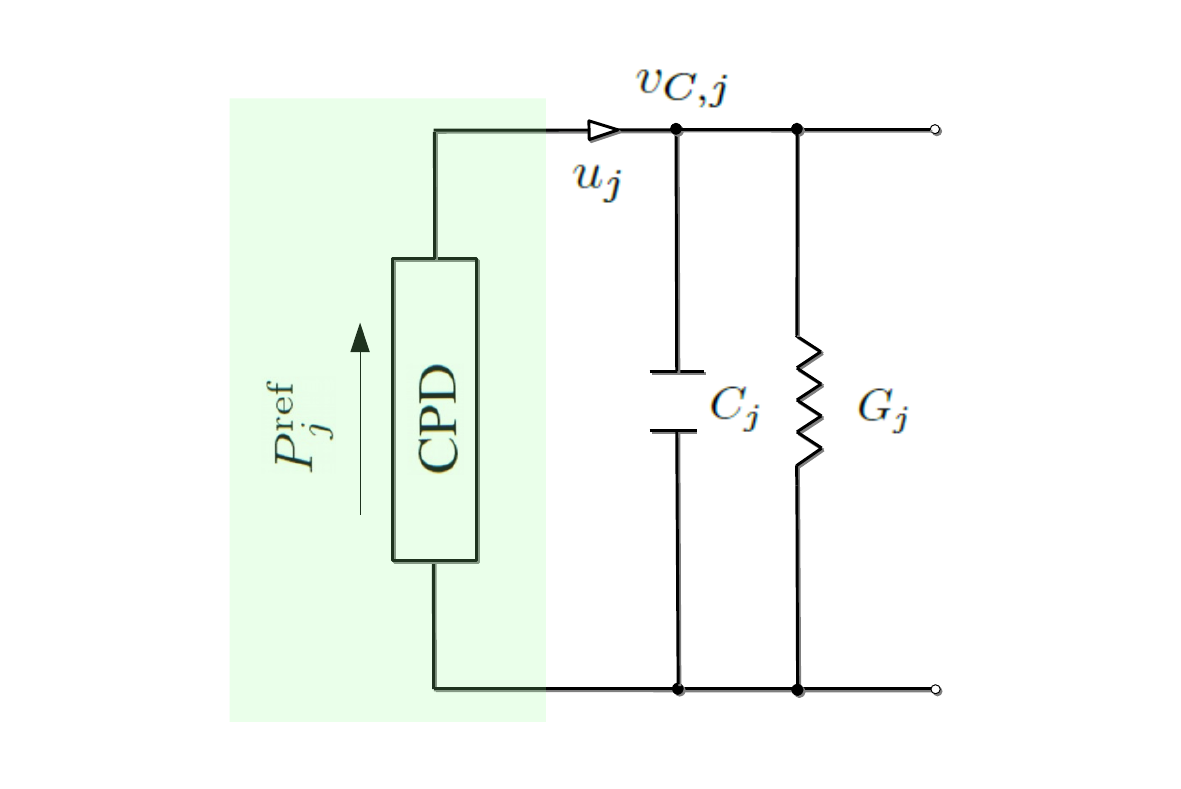}
\caption{Equivalent circuit scheme for \textit{PQ units}.}
 \label{approx1}
\end{subfigure}
\begin{subfigure}{0.44\textwidth}
\includegraphics[width=\linewidth]{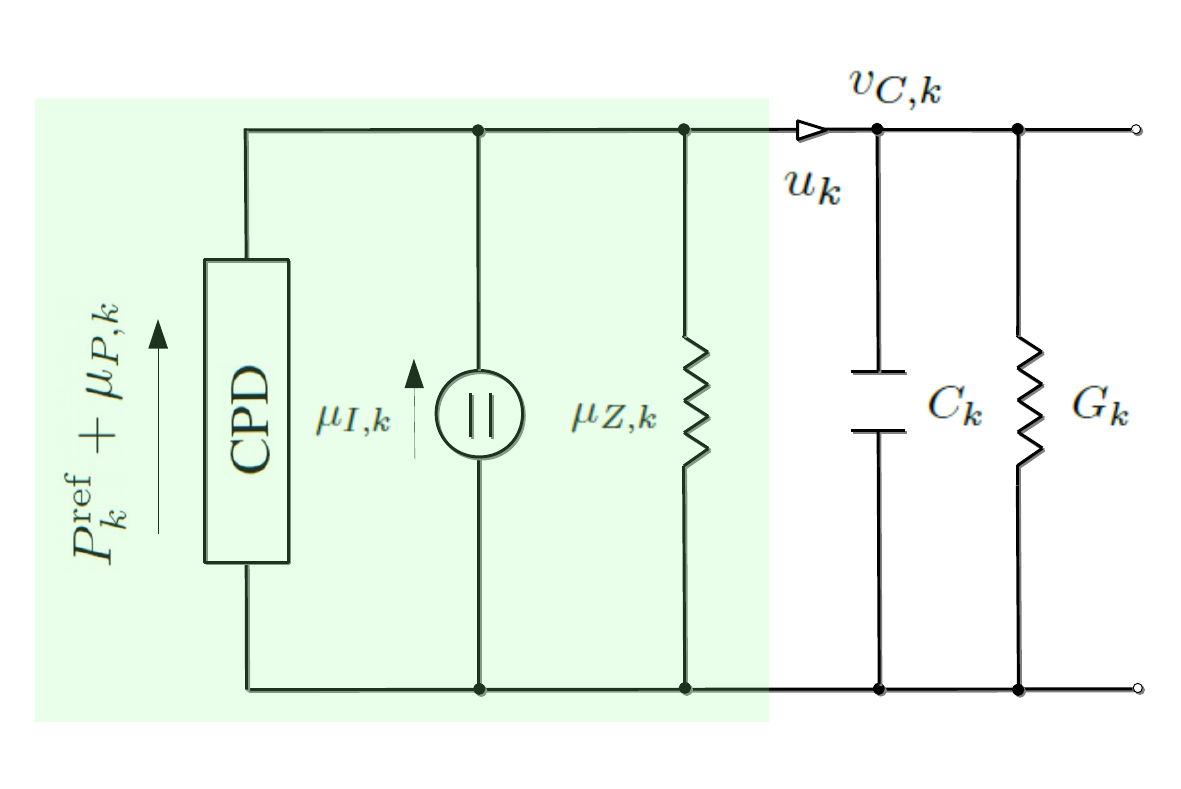}
\caption{Equivalent circuit scheme for \textit{voltage-controlled units}.} \label{approx2}
\end{subfigure}
 \caption{Equivalent circuit schemes of the converter units with constant power devices (CPDs), under Assumption \ref{assCfast}. }
\end{figure*}
Finally, the dynamics of the \textit{PQ units} can be represented by the following scalar systems: 
\begin{equation*}
\begin{aligned}
C_j\dot v_{C,j}&=-G_{j}v_{C,j}+u_{j}+i_{C,j},\\
0&= P^\mathrm{ref}_{{P},j}-v_{C,j}u_{j},
\end{aligned}
\end{equation*}
while for the dynamics of the \textit{voltage-controlled units} we have:
\begin{equation*}
\begin{aligned}
C_k\dot v_{C,k}&=-(G_{k}-\mu_{Z,k})v_{C,k}+\mu_{I,k}+u_{k}+i_{C,k},\\
0&=P^\mathrm{ref}_{{V},k}-v_{C,k}u_{k},
\end{aligned}
\end{equation*}
with $j\sim\mathcal{E}_{{P}},k\sim\mathcal{E}_{{V}}$ and where $v_{C,j},v_{C,k}\in\mathbb{R}_+$ denote the voltages across the capacitors, $i_{C,j},i_{C,k}\in\mathbb{R}$ denote the network currents, $u_{j},u_{k}\in\mathbb{R}$ denote the currents flowing into the constant power devices, $G_{j}\in\mathbb{R}_+$, $G_{k}\in\mathbb{R}_+$, $C_{j}\in\mathbb{R}_+$, $C_{k}\in\mathbb{R}_+$ denote the conductances and capacitances. The aggregated model is then given by: 
\begin{equation}\label{aggregatedC} 
\begin{aligned}
\begin{bmatrix}
C_P\dot v_{{P}}\\C_V\dot v_{{V}}
\end{bmatrix} =-&\begin{bmatrix}
G_{{P}}&0\\
0&G_{{V}}+G_{Z}
\end{bmatrix}\begin{bmatrix}
v_{{P}}\\v_{{V}}
\end{bmatrix}+ \begin{bmatrix}
u_{{P}}\\u_{{V}}
\end{bmatrix}+
\begin{bmatrix}
0\\\bar u_V
\end{bmatrix}+ \begin{bmatrix}
i_{{P}}\\i_{V}
\end{bmatrix} ,
\end{aligned}
\end{equation}
together with the algebraic constraints:
 \begin{equation*}
 \begin{aligned}
  P^\mathrm{ref}_{{P},j}&=v_{{P},j}u_{{P},i},\quad  P^\mathrm{ref}_{{V},k} =v_{{V},k} u_{{V},k},
  \end{aligned}
\end{equation*} 
with $i\sim\mathcal{E}_{{P}}, k\sim\mathcal{E}_{{V}} $ and the following definitions.
\begin{itemize}
\item[-] State vectors
$$v_{{P}}: =\col(v_{C,j})\in\mathbb{R}^{\mathsf{p}},\quad v_{{V}}: =\col(v_{C,k})\in\mathbb{R}^{\mathsf{v}}.$$
 
\item[-] Network ingoing currents 
$$i_{{P}}:=\col(i_{C,j})\in\mathbb{R}^{\mathsf{p}},\quad i_{{ V}}:=\col(i_{C,k})\in\mathbb{R}^{\mathsf{v}}.$$ 
\item[-] Units ingoing currents 
$$ u_{{P}}:=\col(u_{j})\in\mathbb{R}^{\mathsf{p}},\quad u_{{V}}:=\col(u_{k})\in\mathbb{R}^{\mathsf{v}}.$$ 
\item[-] External sources $\bar u_V:=\col(\mu_{I,k})\in\mathbb{R}^{\mathsf{v}}$.
 \item[-] Matrices
 \begin{equation*}
\begin{aligned} 
C_{{P}}:&=\mathrm{diag}\{C_{j}\}\in\mathbb{R}^{\mathsf{p}\times\mathsf{p}},\;\; C_{{V}}:=\mathrm{diag}\{C_{k}\}\in\mathbb{R}^{\mathsf{v}\times\mathsf{v}},\\
G_{{P}}:&=\mathrm{diag}\{G_{j}\}\in\mathbb{R}^{\mathsf{p}\times\mathsf{p}},\;\;  G_{{V}}:=\mathrm{diag}\{G_{k}\}\in\mathbb{R}^{\mathsf{v}\times\mathsf{v}},\;\;
G_Z:=\mathrm{diag}\{-\mu_{Z,k}\}\in\mathbb{R}^{\mathsf{v}\times\mathsf{v}}.\;\;
\end{aligned}
\end{equation*}
\end{itemize}
  \subsection{Interconnected model}
For the model derivation of the hvdc network we assume that the dc transmission lines can be described by standard, single-cell $\pi$-models. However, it should be noted that at each converter node the line capacitors will result in a parallel connection with the output capacitor of the converter \cite{SHAetal}. Hence, the capacitors at the dc output of the converter can be replaced by equivalent capacitors and the transmission lines described by simpler $RL$ circuits, \bk for which it is straightforward to obtain the aggregated model \cite{zonettiCEP}:
\begin{equation}\label{linesrecall}
\begin{aligned}
L_T\dot i_{T} &=-R_{T}i_{T}+v_{T},
\end{aligned}
\end{equation}
with $i_T:=\mathrm{col}(i_{T,i})\in\mathbb{R}^\mathsf{t}$, $v_T:=\mathrm{col}(v_{T,i})\in\mathbb{R}^\mathsf{t}$ denoting the currents through and the voltages across the lines and $L_T:=\mathrm{col}(L_{T,i})\in\mathbb{R}^{\mathsf{t}\times\mathsf{t}}$, $R_T:=\mathrm{col}(R_{T,i})\in\mathbb{R}^{\mathsf{t}\times\mathsf{t}}$ denoting the inductance and resistance matrices.  In order to obtain the reduced, interconnected model of the hvdc transmission system under Assumption \ref{assCfast}, we need  to consider the interconnection laws determined by the incidence matrix \eqref{incidencedc2}. Let us define the node and edge vectors:
\begin{equation*}
V_n:=\begin{bmatrix}
\mathcal{V}_{{P}}\\\mathcal{V}_{{V}} \\0
\end{bmatrix}\in\mathbb{R}^{\mathsf{c}+1},\;
V_e:=\begin{bmatrix}
v_{{P}}\\v_{{V}}\\v_T
\end{bmatrix}\in\mathbb{R}^\mathsf{m},\;
I_e:=\begin{bmatrix}
i_{{P}}\\i_{{V}}\\i_T
\end{bmatrix}\in\mathbb{R}^{ \mathsf{m}}.
\end{equation*}
By using the definition of the incidence matrix \eqref{incidencedc2} together with the Kirchhoff's current  and voltage  laws given by \cite{circuitavdsSCL,zonettiejc}: 
\begin{equation*}
\mathcal{B}I_e=0,\qquad V_e=\mathcal{B}^\top V_n,
\end{equation*}
  we obtain:
\begin{equation}\label{ipiv}
\begin{aligned}
i_{{P}}&=- \mathcal{B}_{{P}}  v_{{P}},\quad
i_{{V}} =-\mathcal{B}_{{V}}  v_{{V}},\quad v_T=\mathcal{B}_{{P}}^\top v_P+\mathcal{B}_{{V}}^\top v_V.
\end{aligned}
\end{equation}
Replacing $i_P$ and $i_V$ in \eqref{aggregatedC} and $v_T$ in \eqref{linesrecall}, leads to the interconnected model: 
\begin{equation} \label{overallred}
\begin{aligned}
 \begin{bmatrix}
C_P\dot v_{{P}}\\C_V\dot v_{{V}}\\L_T\dot i_T
\end{bmatrix}=&\begin{bmatrix}
-G_{{P}}&0&- \mathcal{B}_{{P}} \\
0& -G_{{V}}&- \mathcal{B}_{{V}} \\
\mathcal{B}_{{P}}^\top&\mathcal{B}_{{P}}^\top&-G_Z
\end{bmatrix}\begin{bmatrix}
v_{{P}}\\v_{{V}}\\i_T
\end{bmatrix}+ \begin{bmatrix}
u_{{P}}\\u_{{V}}\\0
\end{bmatrix}+
\begin{bmatrix}
0\\\bar u_V\\0
\end{bmatrix},
\end{aligned}
\end{equation}
together with the algebraic constraints:
 \begin{equation}\label{cplred}
 \begin{aligned}
  P^\mathrm{ref}_{{P},j}=v_{{P},j}u_{{P},j},\quad   P^\mathrm{ref}_{{V},k}=v_{{V},k} u_{{V},k},
  \end{aligned}
\end{equation} 
with $i\sim\mathcal{E}_{{P}}, k\sim\mathcal{E}_{{V}}$.
\begin{remark}\label{drooprem}   With the choice 
$$\mu_{P,k}=0,\quad \mu_{I,k}= d_kv_{C}^\mathrm{nom},\quad \mu_{Z,k}=- d_k,$$
the primary control  \eqref{refsZIP} reduces to:
\begin{equation*}
\delta_k(v_{C,k} )=-\frac{d_k}{V_{d,k}^\star}(v_{C,k}  -v_{C}^\mathrm{nom}),
\end{equation*}
while the injected current is simply given by
$$
i_k^\star=\frac{V_{d,k}^\star}{v_{C,k}}I_{d,k}^\star=\frac{P_k^\mathrm{ref}}{v_{C,k}} - d_k(v_C^\mathrm{nom}-v_{C,k}),
$$
 with $k\sim\mathcal{E}_V$. This is exactly the conventional, widely diffused, voltage droop control \cite{araujo,haile,vrana}, where $d_k$ is called droop coefficient and $v^\mathrm{nom}_C$ is the nominal voltage of the hvdc system. The conventional droop control can be interpreted as an appropriate parallel connection of a current source with an impedance, which is put in parallel with a constant power device, thus resulting in a ZIP model. A similar model is encountered in \cite{Beerten2013} and should  be contrasted with the models provided in \cite{andreasson,zhao}, where the contribution of the constant power device is absent. 
\end{remark}
\smallbreak
\begin{remark}\label{gendc}
A peculiarity of hvdc transmission systems with respect to generalized dc grids is the absence of traditional loads. Nevertheless, the aggregated model of the converter units  \eqref{aggregatedC} can be still employed for the modeling of dc grids with no loss of generality, under the assumption that loads can be represented either by \textit{PQ units} (constant power loads) or by \textit{voltage-controlled units} with assigned parameters (ZIP loads). This model should be contrasted with the linear models adopted in \cite{zhao,andreasson} for dc grids, where loads are modeled as constant current sinks.
\end{remark}
\section{A REDUCED MODEL FOR GENERAL DC SYSTEMS WITH SHORT LINES CONFIGURATIONS}\label{sec_short} 
Since hvdc transmission systems are usually characterized by very long, \textit{i.e.} dominantly inductive, transmission lines, there is no clear time-scale separation between the dynamics of the power converters and the dynamics of the hvdc network. This fact should be contrasted with traditional power systems---where a time-scale separation typically holds because of the very slow dynamics of generation and loads compared to those of transmission lines \cite{sauer}---and microgrids---where a time-scale separation is justified by the short length, and consequently fast dynamics, of the lines \cite{schiffer2016}.   Nevertheless, as mentioned in Remark \ref{gendc}, the model \eqref{overallred}-\eqref{cplred} is suitable for the description of a very general class of dc grids. By taking this model as a point of departure, we thus introduce a reduced model that is particularly appropriate for the description of a special class of dc grids, \textit{i.e.} dc grids with short lines configurations. This class includes, among the others, the widely popular case of dc microgrids \cite{akagidc} and the case of hvdc transmission systems with back-to-back configurations \cite{bucher}. \bk For these configurations, we can then make the following assumption. \bk \smallbreak

\begin{assumption}\label{assTfast} The dynamics of the dc transmission lines  evolve on a time-scale that is much faster than the time-scale at which the dynamics of the voltage capacitors evolve.
\end{assumption}
\smallbreak

Under Assumption \ref{assTfast}, \eqref{linesrecall} reduces to: 
\begin{equation}\label{linesred}
i_T\equiv i_T^\star=G_T v_T,
\end{equation} 
where $i_T^\star$ is the steady-state vector of the line currents and $G_T:=R_T^{-1}$ the conductance matrix of the transmission lines. By replacing the expression \eqref{linesred} into \eqref{overallred} we finally obtain:
\begin{equation}\label{reduced}
\begin{aligned}
 \begin{bmatrix}
C_P\dot v_{{P}}\\C_V\dot v_{{V}}
\end{bmatrix}=-&\begin{bmatrix}
\mathcal{L}_{{P}}+G_{{P}}&\mathcal{L}_m\\
\mathcal{L}_m^\top&\mathcal{L}_{{V}}+G_{{V}}+G_Z
\end{bmatrix}\begin{bmatrix}
v_{{P}}\\v_{{V}}
\end{bmatrix}+ \begin{bmatrix}
u_{{P}}\\u_{{V}}
\end{bmatrix}+
\begin{bmatrix}
0\\\bar u_V
\end{bmatrix},
\end{aligned}
\end{equation}
together with the algebraic constraints \eqref{cplred} and where we defined
\begin{equation*}
\begin{aligned}
\mathcal{L}_P:&=\mathcal{B}_{{P}}G_L\mathcal{B}_{{P}}^\top,\quad\mathcal{L}_m:=\mathcal{B}_{{P}}G_L\mathcal{B}_{{V}}^\top,\quad
\mathcal{L}_V:=\mathcal{B}_{{V}}G_L\mathcal{B}_{{V}}^\top.
\end{aligned}
\end{equation*}
\begin{remark} 
The matrix:
\begin{equation*}
\mathcal{L}:=\begin{bmatrix}
\mathcal{L}_{{P}} &\mathcal{L}_m\\
\mathcal{L}_m^\top&\mathcal{L}_{{V}}
\end{bmatrix}\in\mathbb{R}^{\mathsf{c}\times\mathsf{c}}
\end{equation*}
is the Laplacian matrix associated to the weighted undirected graph $\bar{\mathcal{G}}^w$, obtained from the (unweighted directed) graph $\mathcal{G}^\uparrow$ that describes the hvdc transmission system by: 1) eliminating the reference node and all edges connected to it; 2) assigning as weights of the edges corresponding to transmission lines the values of their conductances. Similar definitions are also encountered in \cite{andreasson,zhao}.   
 \end{remark}

\section{CONDITIONS FOR EXISTENCE OF AN EQUILIBRIUM POINT}\label{sec_existence}
 From an electrical point of view, the reduced system \eqref{overallred}-\eqref{cplred} is a linear  $RLC$    circuit, where at each node a constant power device is attached.  It has been observed in experiments and simulations that the presence of constant power devices may seriously affect the dynamics of    these    circuits hindering the achievement of a constant, stable behavior of the state variables---the dc voltages in the present case \cite{cplcooley,kwasinski11,barabanov,sanchez}. A first objective is thus to determine  conditions on the free control parameters of the system \eqref{overallred}-\eqref{cplred}     for   the existence of an equilibrium point. 
   Before presenting the main result of this section, we make an important observation: since the steady-state of the system \eqref{overallred}-\eqref{cplred} is equivalent to the steady-state of the system \eqref{reduced}-\eqref{cplred}, the analysis of existence of an equilibrium point follows \textit{verbatim}. Based on this consideration, in the present section we will only consider the system \eqref{reduced}-\eqref{cplred}, bearing in mind the the same results hold for the system \eqref{overallred}-\eqref{cplred}. \bk  
 To simplify the notation, we define
\begin{equation}\label{rprv}
\begin{aligned}
P_{{P}}^\mathrm{ref}:&=\mathrm{col}(P_{{P},j}^\mathrm{ref})\in\mathbb{R}^\mathsf{p},\quad R_P: =\mathcal{L}_P+G_P\in\mathbb{R}^{\mathsf{p}\times\mathsf{p}},\\
P_{{V}}^\mathrm{ref}:&=\mathrm{col}(P_{{V},k}^\mathrm{ref})\in\mathbb{R}^\mathsf{v}, \quad R_V: =\mathcal{L}_V+G_V+G_Z\in\mathbb{R}^{\mathsf{v}\times\mathsf{v}}.
\end{aligned}
\end{equation}
Furthermore, we   recall the following lemma, the proof of which can be found in \cite{barabanov}.\smallbreak

\begin{lemma}\label{lemma} Consider $m$ quadratic equations of the form $f_i:\mathbb{R}^n\rightarrow \mathbb{R}$,
\begin{equation}\label{LMIlemma}
f_i(x):=\frac{1}{2}x^\top\mathcal{A}_i x+x^\top\mathcal{B}_i,\qquad i\in[1,m],
\end{equation}
where $\mathcal{A}_i=\mathcal{A}_i^\top\in\mathbb{R}^{n\times n}$, $\mathcal{B}_i\in\mathbb{R}^n$, $  c_i\in\mathbb{R} $ and define:
\begin{equation*} 
\begin{aligned} 
\mathcal{A}(T):&=\sum_{i=1}^mt_i\mathcal{A}_i,\quad  \mathcal{B}(T):=\sum_{i=1}^mt_i\mathcal{B}_i,\quad
\mathcal{C}(T): =\sum_{i=1}^m t_ic_i.
\end{aligned} 
\end{equation*}
If the following LMI
\begin{equation*}
\Upsilon(T):=\begin{bmatrix}
\mathcal{A}(T)&\mathcal{B}(T)\\
\mathcal{B}^\top(T)&-2\mathcal{C}(T)
\end{bmatrix}>0,
\end{equation*}
is feasible, then the equations 
\begin{equation}\label{finP}
f_i(x)=c_i,\qquad i\in[1,m],
\end{equation}
have no solution.
\end{lemma}
\smallbreak
We are now ready to formulate the following proposition, that establishes necessary, control parameter-dependent, conditions for the existence of equilibria of the system \eqref{reduced}-\eqref{cplred}.\smallbreak

\begin{proposition}\label{LMIprop} Consider the system \eqref{reduced}-\eqref{cplred}, for some given $P_{{P}}^\mathrm{ref}\in\mathbb{R}^\mathsf{p}$, $P_{{V}}^\mathrm{ref}\in\mathbb{R}^\mathsf{v}$. Suppose that there exist  two diagonal matrices $T_{{P}}\in\mathbb{R}^{\mathsf{p}\times\mathsf{p}}$ and $T_{{V}}\in\mathbb{R}^{\mathsf{v}\times\mathsf{v}}$ such that:
\begin{equation}\label{LMI}
\Upsilon(T_P,T_V)>0,
\end{equation}
with
\begin{equation*}
\Upsilon :=\begin{bmatrix}
T_{{P}} R_P+R_PT_{{P}}&T_{{P}}\mathcal{L}_m+\mathcal{L}_m^\top T_{{V}}&0\\
\star &T_{{V}} R_V +  R_V T_{{V}}&-T_{{V}}\bar u_V\\
\star&\star&-2( \mathsf{1}_\mathsf{p}^\top T_P  P_{{P}}^\mathrm{ref}+\mathsf{1}_\mathsf{v}^\top T_V P_{{V}}^\mathrm{ref})
\end{bmatrix},
\end{equation*}
where $P_P^\star$, $P_V^\star$, $R_P$ and $R_V$ are defined in \eqref{rprv}. Then  the system \eqref{reduced}-\eqref{cplred}  does not admit an equilibrium point.
\end{proposition}
\begin{proof}
First of all, by setting the left-hand of the differential equations in \eqref{reduced} to zero and using \eqref{rprv}, we have:
\begin{equation*} 
\begin{aligned}
0=&- R_P  v_{{P}}^\star-\mathcal{L}_mv_{{V}}^\star+u_{{P}}^\star, \\
0=&-\mathcal{L}_m^\top v_{{P}}^\star- R_V v_{{V}}^\star+u_{{V}}^\star+\bar u_V.
\end{aligned} 
\end{equation*}
Left-multiplying  the first and second set of equations by $v_{{P},j}^\star$ and $v_{{V},k}^\star$ respectively, with  $j\sim\mathcal{E}_{{P}}$, $k\sim\mathcal{E}_{{V}}$, we get
\begin{equation*} 
\begin{aligned}
P_{{P},j}^\mathrm{ref}&=v_{{P},j}^\star R^\top_{P,j} v_{{P}}^\star+v_{{P},j}^\star\mathcal{L}^\top_{m,j}v_{{V}}^\star,\\
P_{{V},k}^\mathrm{ref}&=v_{{V},k}^\star \mathcal{L}_{m,k}  v_{{P}}^\star+v_{{V},k}^\star R_{V,k}^\top v_{{V}}^\star-v_{{V},k}^\star\bar u_{V,k},
\end{aligned} 
\end{equation*}
which, after some manipulations, gives
\begin{equation}\label{quadratic}
\begin{aligned}
c_i&=\frac{1}{2}(v^\star)^\top\mathcal{A}_iv^\star+(v^\star)^\top\mathcal{B}_i,
\end{aligned}
\end{equation}
with $i\sim\mathcal{E}_{P}\cup\mathcal{E}_V$, $v^\star:=\mathrm{col}(v_{{P}}^\star,v_{{V}}^\star)\in\mathbb{R}^\mathsf{c}$ and
\begin{equation*} 
\begin{aligned}
 \mathcal{A}_i:&= e_ie_i^\top\begin{bmatrix}
R_P&\mathcal{L}_m\\
\mathcal{L}_m^\top&R_V
\end{bmatrix}+\begin{bmatrix}
R_P&\mathcal{L}_m^\top\\
\mathcal{L}_m&R_V
\end{bmatrix}e_ie_i^\top,\quad
\mathcal{B}_i:=e_ie_i^\top \begin{bmatrix}
0\\\bar u_V
\end{bmatrix},\quad
c_i :=e_i^\top \begin{bmatrix}P_{P}^\mathrm{ref}\\ P_V^\mathrm{ref} \end{bmatrix}.
\end{aligned}
\end{equation*}
 Let consider the map $f(v^\star):\mathbb{R}^\mathsf{c}\rightarrow\mathbb{R}^\mathsf{c}$ with components
$$
f_i(v^\star)= \frac{1}{2}(v^\star)^\top\mathcal{A}_iv^\star,
$$
with $i\sim\mathcal{E}_{P}\cup\mathcal{E}_{V}$ and denote by $F$ the image of $\mathbb{R}^{\mathsf{c}}$ under this map.  The problem of solvability of such equations can be formulated as  in Lemma \ref{lemma}, \textit{i.e.} if the LMI \eqref{LMI} holds, then $\mathrm{col}(c_i^\star)$ is not in $F$, thus completing the proof.
\end{proof}
 \smallbreak
 
\begin{remark}  Note that the feasibility of the LMI \eqref{LMI} depends   on the system topology reflected in the Laplacian matrix $\mathcal{L}$ and \bk on the system parameters, among which $G_Z$, $\bar u_V$ and $P_{{V}}^\mathrm{ref}$ are free (primary) control parameters. Since the feasibility condition is only necessary for the existence of equilibria for \eqref{overallred}, it is of interest to determine regions for these parameters that imply non-existence of an equilibrium point.
\end{remark}
 \section{CONDITIONS FOR POWER SHARING}\label{sec_powersharing}
As already discussed, another control objective of primary control is the achievement of \textit{power sharing} among the\textit{ voltage-controlled units}. This property consists in guaranteeing an appropriate (proportional) power distribution among these units in steady-state. We next show that is possible to reformulate such a control objective as a set of quadratic constraints on the equilibrium point, assuming that it exists.   Since it is a steady-state property, the same observation done in Section \ref{sec_existence} applies, which means that the results obtained for the system \eqref{reduced}-\eqref{cplred} also hold for the system \eqref{overallred}-\eqref{cplred}. We introduce the following definition.\smallbreak

 \begin{definition}\label{psdef} Let be   \mbox{$v^\star:=(v_{{P}}^\star,v_{{V}}^\star)\in\mathbb{R}^\mathsf{c}$} and $P_{\mathrm{DC},V}(v^\star):=\mathrm{col}({P}_{\mathrm{DC},k}(v^\star_{C,k}))\in\mathbb{R}^{\mathsf{v}}$ respectively an equilibrium point for the system \eqref{reduced}-\eqref{cplred} and the collection of injected powers as defined by \eqref{DCinjected}, and let be $\Gamma:=\mathrm{diag}\{\gamma_k\}\in\mathbb{R}^{\mathsf{v}\times\mathsf{v}},$ a  positive definite matrix. Then $v^\star$ is said to possess the power sharing property with respect to $\Gamma$ if:
 \begin{equation}\label{ps}
\Gamma P_{\mathrm{DC},V}(v^\star) =\mathsf{1}_{\mathsf{v}}.
\end{equation}
\end{definition}
\smallbreak

Then we have the following lemma.\smallbreak
 
\begin{lemma}\label{lemma2} Let  $v^\star=(v_{{P}}^\star,v_{{V}}^\star)\in\mathbb{R}^\mathsf{c}$ be an equilibrium point for \eqref{reduced}-\eqref{cplred} and $\Gamma:=\mathrm{diag}\{\gamma_k\}\in\mathbb{R}^{\mathsf{v}\times\mathsf{v}}$ a  positive definite matrix. Then $v^\star$ possesses the power sharing property with respect to $\Gamma$ if an only if the quadratic equations
 \begin{equation}\label{quadraticps}
\frac{1}{2}(v^\star)^\top\mathcal{A}^{\mathrm{ps}}_kv^\star+(\mathcal{B}^{\mathrm{ps}}_k)^\top v^\star=p^{\mathrm{ps}}_k,
\end{equation}
with $ k\sim\mathcal{E}_{V}$ and where:
\begin{align*}
\mathcal{A}^{\mathrm{ps}}_k:&= 2\begin{bmatrix}0&0\\0&\Gamma G_Z\end{bmatrix}e_{k}e_{k}^\top ,\quad 
 \mathcal{B}^{\mathrm{ps}}_k:  =\begin{bmatrix}0\\\Gamma\bar u_V\end{bmatrix}e_ke_k^\top,\quad
  p^{\mathrm{ps}}_k:=e_k^\top\begin{bmatrix}
  0\\\Gamma P_{{V}}^\mathrm{ref},
\end{bmatrix}   
\end{align*}
admit  a solution.
\end{lemma}
\begin{proof}
From \eqref{ps}  we have that by definition:
\begin{equation*}
\gamma_k P^\mathrm{ref}_{\mathrm{DC},k}(v_{C,k})=1,
\end{equation*}
with $ k\sim\mathcal{E}_{V}$, which by recalling \eqref{DCinjected}, is equivalent to:
\begin{equation*}
\gamma_k(P_{{V},k}^\mathrm{ref}+\mu_{I,k}v_{C,k}+\mu_{Z,k}v_{C,k}^2)=1.
\end{equation*}
After some straightforward manipulations, the above equalities can be rewritten as \eqref{quadraticps}, completing the proof.
\end{proof}
\smallbreak

An immediate implication of this lemma is given in the following proposition, which establishes necessary conditions for the existence of an equilibrium point that verifies the power sharing property.\smallbreak

\begin{proposition} Consider the system \eqref{reduced}-\eqref{cplred}, for some given $P_{{P}}^\mathrm{ref}$, $P_{{V}}^\mathrm{ref}$ and $\Gamma$. Suppose that there exist  three  diagonal matrices $T_{{P}}\in\mathbb{R}^{\mathsf{p}\times\mathsf{p}}$, $T_{{V}}\in\mathbb{R}^{\mathsf{v}\times\mathsf{v}},T_V^{\mathrm{ps}}\in\mathbb{R}^{\mathsf{v}\times\mathsf{v}}$,    such that:
\begin{equation}\label{LMIps}
\Upsilon(T_P,T_V)+\Upsilon_{\mathrm{ps}}(T_V^{\mathrm{ps}})
>0,
\end{equation}
with
\begin{equation*}\Upsilon_{\mathrm{ps}}:=
\begin{bmatrix}
0&0&0\\
\star&2T_V^{\mathrm{ps}}\Gamma G_Z&T_V^{\mathrm{ps}}\Gamma\bar u_V\\
\star&\star&-2T_V^{\mathrm{ps}}(\mathsf{1}_\mathsf{v} -\Gamma P_{{V}}^\mathrm{ref})
\end{bmatrix}.
\end{equation*}
Then the system \eqref{reduced}-\eqref{cplred} does not admit an equilibrium point that verifies the power sharing property.
\end{proposition}
\begin{proof}
The proof is similar to the proof of Proposition \ref{LMIprop}. By using Lemma \ref{lemma2} the power sharing constraints can be indeed  rewritten as quadratic equations, similarly to \eqref{quadratic}. Hence, it suffices to apply Lemma \ref{lemma} to the quadratic equations \eqref{quadratic}, \eqref{quadraticps}  to complete the proof.
\end{proof}

\section{CONDITIONS FOR LOCAL ASYMPTOTIC STABILITY}\label{sec_stability}
We now present a result on stability of a given equilibrium point for the system \eqref{overallred}-\eqref{cplred}. The result is obtained by applying Lyapunov's first method.\smallbreak

\begin{proposition} Consider the system \eqref{overallred}-\eqref{cplred} and assume that $v^\star=(v_{{P}}^\star,v_{{V}}^\star,i_T^\star)\in\mathbb{R}^\mathsf{m}$ is an equilibrium point. Let 
\begin{equation}\label{defG}
\begin{aligned}
G_{{P}}^\star:&=\mathrm{diag} \left\lbrace \frac{P_{{P},j}^\mathrm{ref}}{(v_{{P},j}^\star)^2}\right\rbrace \in\mathbb{R}^{\mathsf{p}\times\mathsf{p}},\quad
G_{{V}}^\star:=\mathrm{diag}\left\lbrace \frac{P_{{V},k}^\mathrm{ref}}{(v_{{V},k}^\star)^2}\right\rbrace \in\mathbb{R}^{\mathsf{v}\times\mathsf{v}},
\end{aligned}
\end{equation}
and
\begin{equation*}
J(v^\star):=-\begin{bmatrix}
-C_{{P}}^{-1}(G_P+G_{{P}}^\star)&0&-C_{{P}}^{-1}\mathcal{B}_P\\
0&-C_{{V}}^{-1}(G_V+G_{{V}}^\star)&-C_V^{-1}\mathcal{B}_V\\
L_T^{-1}\mathcal{B}_P^\top&L_T^{-1}\mathcal{B}_V^\top&-L_T^{-1}R_T
\end{bmatrix}. 
\end{equation*}
 Then if:
\begin{itemize}
\item[-]  all eigenvalues $\lambda_i$ of $J$ are such that 
$$\mathfrak{Re} \{ \lambda_i \left[J(v^\star)\right]\}<0,$$
 the equilibrium point  $ v^\star $ is locally asymptotically stable;
\item[-]  there exists at least one eigenvalue $\lambda_i$ of $J$ such that 
$$\mathfrak{Re} \{ \lambda_i \left[J(v^\star)\right]\}>0,$$
  the equilibrium point  $ v^\star$  is unstable.
\end{itemize}
\end{proposition}
\begin{proof}
The first-order approximation of the system \eqref{overallred}-\eqref{cplred} around  $v^\star$ is given by:
\begin{equation}\label{jacobian}
\begin{aligned}
\begin{bmatrix}
C_P\dot v_{{P}}\\C_V\dot v_{{V}}\\L_T\dot i_T
\end{bmatrix}  =&\begin{bmatrix}
-G_P& 0&-\mathcal{B}_P\\
0& -G_V &-\mathcal{B}_V\\
\mathcal{B}_P^\top&\mathcal{B}_V^\top &-R_T \end{bmatrix}\begin{bmatrix}
v_{{P}}\\v_{{V}}\\i_T
\end{bmatrix} +\begin{bmatrix}
 \frac{\partial i_{{P}} }{\partial v_{{P}}}\big\vert_{v^\star}& 0&0\\
0& \frac{\partial i_{{V}} }{\partial v_{{V}}}\big\vert_{v^\star}&0\\
 0&0&0
\end{bmatrix}\begin{bmatrix}
v_P\\v_V\\i_T
\end{bmatrix}
\end{aligned} 
\end{equation}

Differentiating \eqref{cplred} with respect to $v_{{P}} $, $v_{{V}} $, yields:
\begin{equation*}
\begin{aligned}
0_{\mathsf{p}\times\mathsf{p}}&= \frac{\partial i_{{P}}}{\partial v_{{P}}}\cdot\mathrm{diag}\{v_{{P},j}\}+\mathrm{diag}\{i_{{P},j}\},\quad
0_{\mathsf{v}\times\mathsf{v}}=\frac{\partial i_{{V}}}{\partial v_{{V}}}\cdot\mathrm{diag}\{v_{{V},k}\}+\mathrm{diag}\{i_{{V},k}\}.
\end{aligned}
\end{equation*}
By using \eqref{defG}, it  follows that
\begin{equation*}
\frac{\partial i_{{P}}}{\partial v_{{P}}}\bigg\vert_{v^\star}=-G_{{P}}^\star,\quad \frac{\partial i_{{V}}}{\partial v_{{V}}}\bigg\vert_{v^\star}=-G_{{V}}^\star .
\end{equation*}
The proof is   completed by replacing into \eqref{jacobian} and invoking Lyapunov's first method.
\end{proof}

\section{AN ILLUSTRATIVE EXAMPLE}\label{sec_example}
In order to validate the results on existence of equilibria and power sharing for the system \eqref{overallred}-\eqref{cplred} we next provide an illustrative example. Namely, we consider the four-terminal hvdc transmission system depicted in Fig.~\ref{illudroop}, the parameters of which are given in Table \ref{par4hvdc}. 

\begin{table*}
\caption{System parameters.}
\label{par4hvdc}       
%
%
\centering
\begin{tabular}{p{0.5cm}p{1.5cm}p{0.5cm}p{1.7cm}p{0.5cm}p{1.7cm}p{0.5cm}p{1.7cm}p{0.5cm}p{1.7cm}}
\hline
  & Value &   & Value & & Value &    & Value &  & Value  \\
\hline
 $G_{i}$   & $0$ $\Omega^{-1} $ &  $P_{V,1}^\star$  & $30$ $MW$  & $P_{P,2}^\star$ & $-20$ $MW $ &
         $P_{V,3}^\star$               & $9$ $MW  $ &  $P_{P,4}^\star$ & $-24$ $MW$ \\
$C_{i}$  & $20$ $\mu F $ &  $G_{12}$ &$0.1 $ $\Omega^{-1}  $ & $G_{14}$ & $0.15$ $\Omega^{-1}  $ &
         $G_{23}$               & $0.11$ $\Omega^{-1}  $ &  $G_{24}$ & $0.08$ $\Omega^{-1} $   \\ 
\hline 
\end{tabular}
\end{table*}
  \begin{figure}[h!]
 \centering
 \includegraphics[width=.65\linewidth]{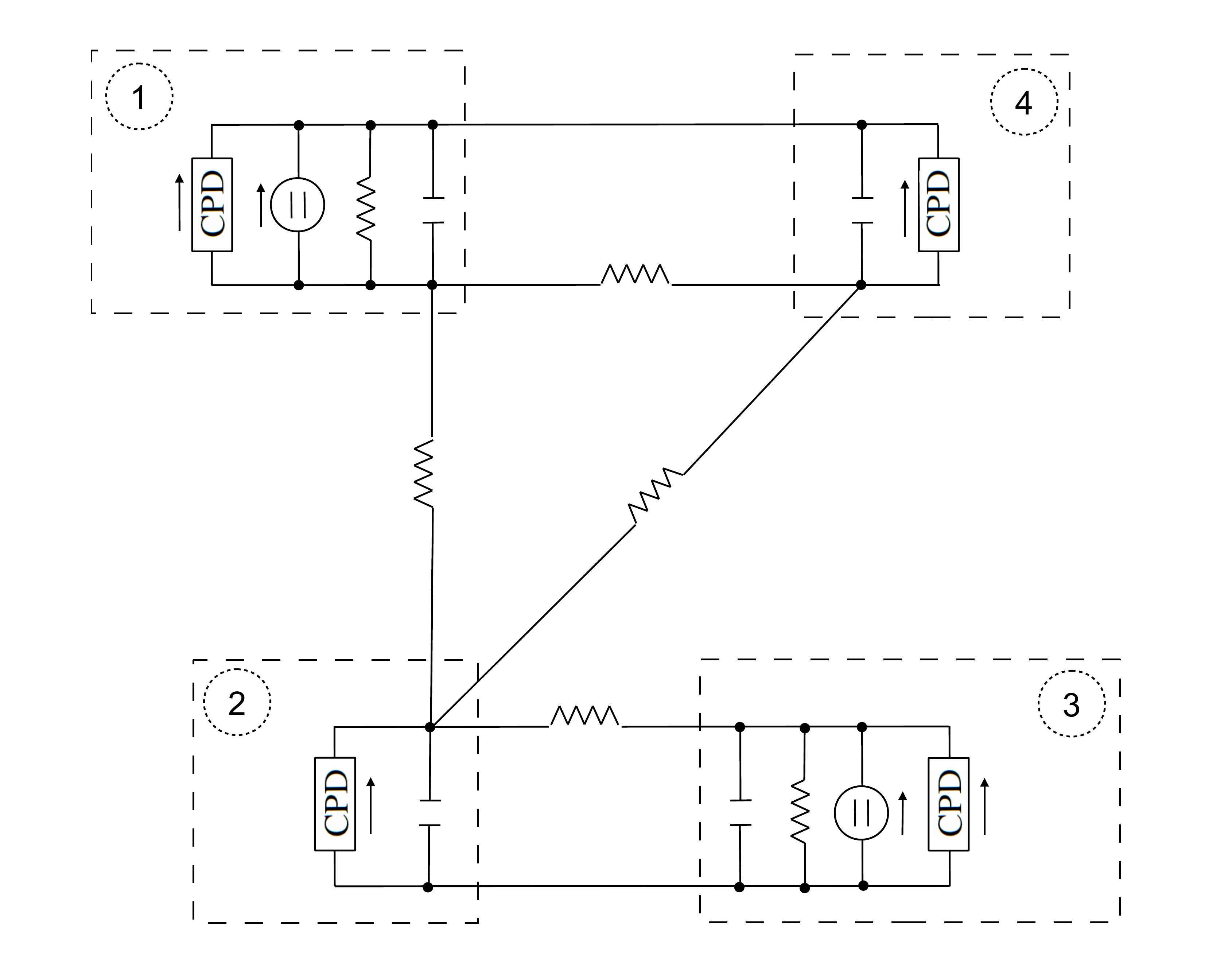}
 \caption{Four-terminal hvdc transmission system. }
 \label{illudroop}
\end{figure}

Since $\mathsf{c}=\mathsf{t}=4$, the graph associated to the hvdc system has $\mathsf{n}=4+1=5$ nodes  and $\mathsf{m}=4+4=8$ edges. We then make the following assumptions.
\begin{itemize}
\item[-] Terminal 1   and Terminal 3   are equipped with primary control, from which it follows that there are $\mathsf{p}=2$ \textit{PQ units} and $\mathsf{v}=2$ \textit{voltage-controlled units}. More precisely we take
$$
\delta_k(v_{C,k})=-\frac{d_k}{V_{d,k}^\star}(v_{C,k}-v_{C}^{\mathrm{nom}}),\quad k=\{1,3\}.
$$
This is the well-known \textit{voltage droop control}, where $d_k$ is a free control parameter, while $v_{C}^\mathrm{nom}$ is the nominal voltage of the hvdc system, see also Remark \ref{drooprem}.
\item[-] The power has to be shared equally among terminal 1 and terminal 3, from which it follows that $\Gamma=\mathbb{I}_2$ in Definition \ref{psdef}.
\end{itemize}

The next results are obtained by investigating the feasibility of the LMIs \eqref{LMI}, \eqref{LMIps} as a function of the free control parameters $d_1$ and $d_3$. For this purpose, CVX, a package for specifying and solving convex programs, has been used to solve the semidefinite programming feasibility problem \cite{cvx}. By using a gridding approach, the regions of the (positive) parameters that guarantee feasibility (yellow) and unfeasibility (blue) of the  LMI \eqref{LMI} are shown in Fig.~\ref{FigLMI}, while in Fig.~\ref{FigLMI2} the same is done with respect to the LMI \eqref{LMIps}. We deduce that a necessary condition for the existence of an equilibrium point is that the control parameters are chosen inside the blue region of Fig.~\ref{FigLMI}. Similarly, a necessary condition  for the existence of an equilibrium point that further possesses the power sharing property is that the control parameters are chosen inside the blue region of Fig.~\ref{FigLMI2}.
\begin{figure} 
 \centering
 \includegraphics[width=.65\linewidth]{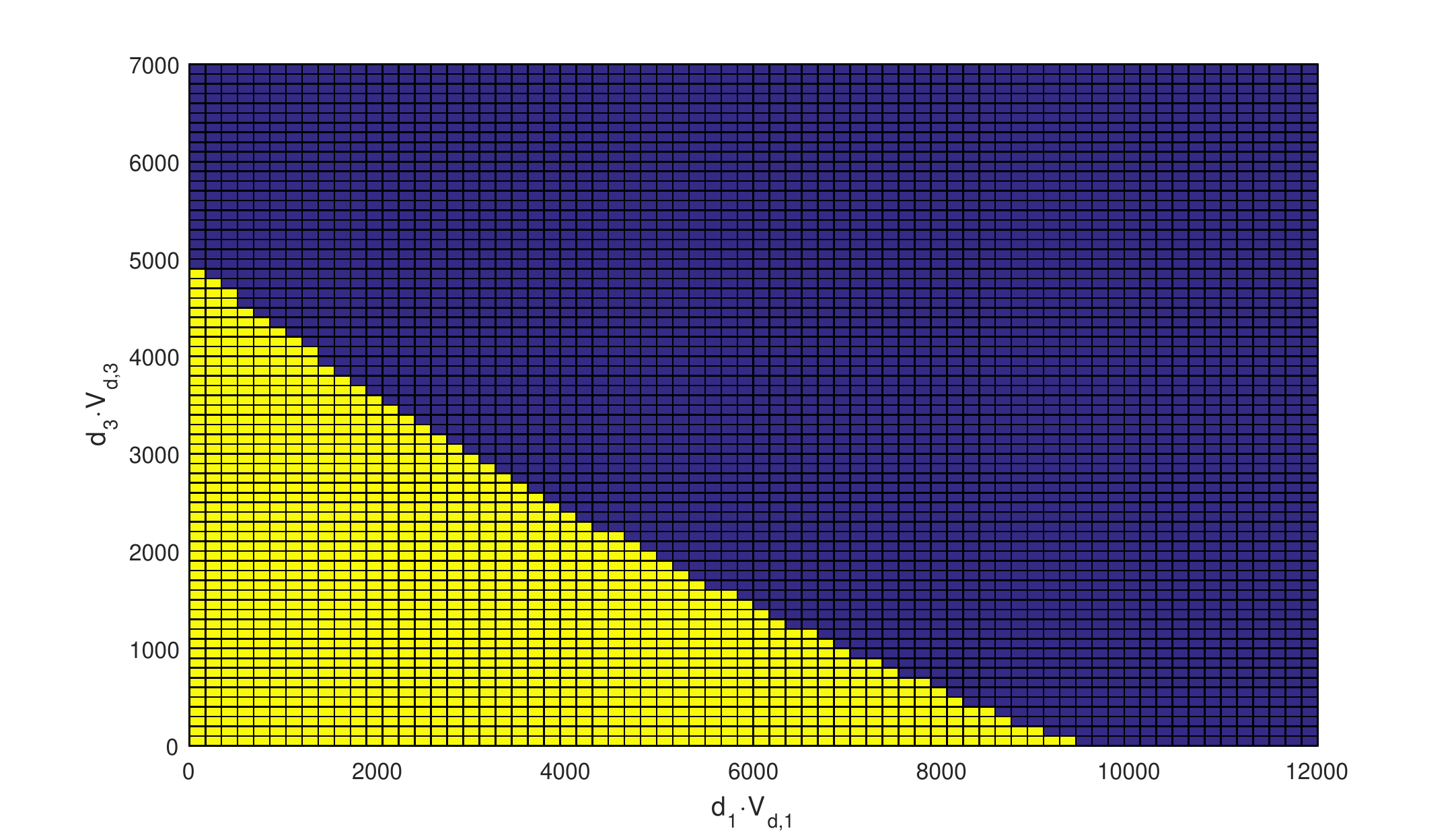}
 \caption{Feasibility regions of the LMI \eqref{LMI} on the plane $(d_1,d_3)$ of droop control parameters. Regions are yellow-coloured if the LMI is feasible and blue-coloured if the LMI is unfeasible.}
 \label{FigLMI}
\end{figure}
\begin{figure} 
 \centering
 \includegraphics[width=.65\linewidth]{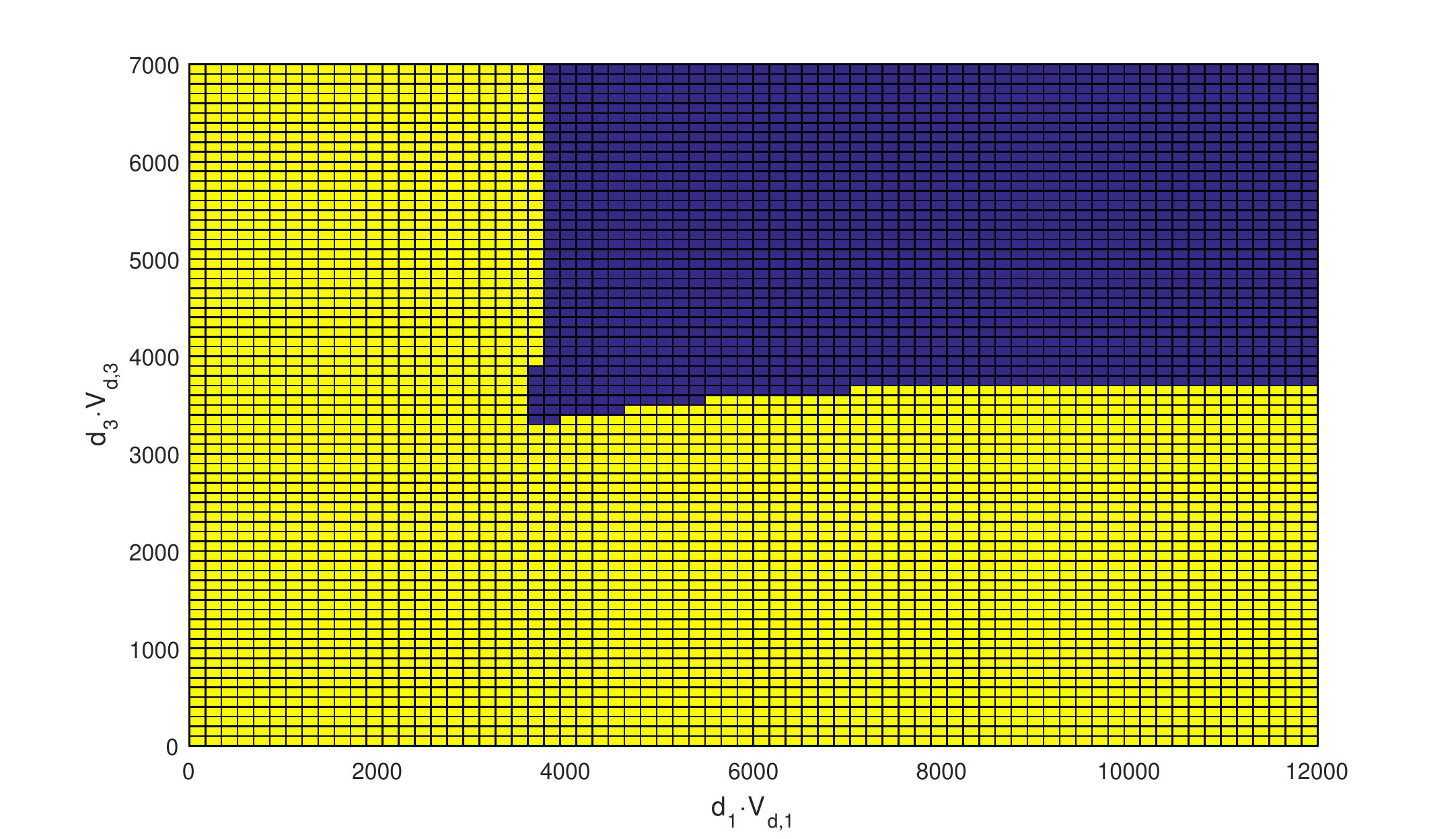}
 \caption{Feasibility regions of the LMI \eqref{LMIps} on the plane $(d_1,d_3)$ of droop control parameters. Regions are yellow-coloured if the LMI is feasible and blue-coloured if the LMI is unfeasible.}
 \label{FigLMI2}
\end{figure}


\section{CONCLUSIONS AND FUTURE WORKS}
 
In this paper, a new nonlinear model for primary control analysis and design has been derived. Primary control laws are described by   equivalent ZIP models, which include the standard voltage droop control  as a special case. A necessary condition for the existence of equilibria in the form of an LMI---which depends on the parameters of the controllers---is established, thus showing that an inappropriate choice of the latter may lead to non-existence of equilibria for the closed-loop system. The same approach is extended to the problem of existence of equilibria that verify a pre-specified power sharing property.   The obtained necessary conditions can be helpful to system operators to tune their controllers such that regions where the closed-loop system will definitely not admit a stationary operating point are excluded. In that regard, the present paper is a first, fundamental stepping stone towards the development of a better understanding of how existence of stationary solutions of hvdc systems are affected by the system parameters, in particular the network impedances and controller gains. \bk A final contribution consists in the establishment of conditions of local asymptotic stability of a given equilibrium point. The obtained results are illustrated on a four-terminal  example. \\

Starting from the obtained model, future research will concern various aspects. First of all, a better understanding of how the feasibility of the LMIs are affected by the parameters is necessary. A first consideration is that the established conditions, besides on the controllers parameters, also depends on the network topology and the dissipation via the Laplacian matrix induced by the electrical network. This suggests that the location of the \textit{voltage-controlled units}, as well as the network impedances, play an important role on the existence of equilibria for the system. Similarly, it is of interest to understand in which measure the values of Z, I and P components of the equivalent ZIP mode affect the LMIs, in order to provide guidelines for the design of primary controllers.   Furthermore, the possibility to combine the obtained necessary conditions with related (sufficient) conditions from the literature, \textit{e.g.} in \cite{bolognanipower}, is very interesting and timely. \bk Other possible developments will focus on the establishment of necessary (possibly sufficient) conditions for the existence of equilibria in different scenarios:  small deviations from the nominal voltage \cite{simpsoncpl,Beerten2013}; power unit outages \cite{Beerten2013}; linear three-phase, ac circuit, investigating the role played by reactive power \cite{sanchez}. 
\section{ACKNOWLEDGMENTS}
The authors acknowledge the support of: the Future Renewable Electric Energy Distribution Management Center (FREEDM), a National Science Foundation supported Engineering Research Center, under grant NSF EEC-0812121; the Ministry of Education and Science of Russian Federation (Project14.Z50.31.0031); the European Union's Horizon 2020 research and innovation programme under the Marie Sklodowska-Curie grant agreement No. 734832.

\end{document}